\pdfoutput=1
\documentclass[conference]{IEEEtran}
\usepackage{times}
\usepackage{epsfig}
\usepackage{subcaption}
\usepackage{amsmath}
\usepackage{amsfonts}
\usepackage{graphicx}
\usepackage{amssymb}
\usepackage{amstext}
\usepackage{latexsym}
\usepackage{color,soul}
\usepackage{ifthen}
\usepackage{multirow,soul}
\usepackage{verbatim}
\usepackage{array,tabularx}
\usepackage[nospace]{cite}
\usepackage[mathscr]{euscript}
\usepackage{tikz}
\usepackage{accents}
 \usepackage{cite}
 \usepackage{bm}
 \usepackage{arydshln}
\usepackage{hyperref}
\usepackage{amssymb,epsfig,graphics, graphicx, url,times,mathrsfs,algorithm,algorithmic,amsmath,array,latexsym,xspace,fancyhdr,wrapfig, xcolor,multirow,amsthm,cite}
\usepackage{caption}
\usepackage{amsmath}
\usepackage{soul}
\usepackage[inline]{enumitem}
\usepackage{arydshln}
\usepackage{enumitem}
\usepackage[makeroom]{cancel}
\usepackage{mathtools}
\usepackage{dblfloatfix}
\usepackage[utf8]{inputenc}
 \usepackage{tikz,pgfplots}
\usetikzlibrary{shapes.geometric}
\usetikzlibrary{shapes,snakes,patterns}
\usetikzlibrary{arrows,positioning,automata,calc}
\usetikzlibrary{plotmarks}
\tikzset{
    int/.style={
           rectangle,
           rounded corners,
           draw=black, thin, fill=black!20,
           minimum height=2em,
           inner sep=2pt,
           text centered,
           },
}

\newtheorem{theorem}{Theorem}

\newtheorem{claim}{Claim}
\newtheorem{example}{Example}
\newtheorem{definition}{Definition}
\newtheorem{remark}{Remark}

\IEEEoverridecommandlockouts

\pdfoutput=1
\begin{document}
\pdfoutput=1
\allowdisplaybreaks[1]
\newlength\figureheight
\newlength\figurewidth

\title{ Guess \& Check Codes for Deletions, Insertions, and Synchronization\vspace{-0.2cm} \thanks{This is an extended version of the conference paper in~\cite{GC}, which was presented at the 2017 IEEE International Symposium on Information Theory.} \thanks{This work was supported  in parts by NSF Grant CCF 15-26875.}}
\author{
\IEEEauthorblockN{Serge Kas Hanna,  Salim El Rouayheb \\ ECE Department, Rutgers University\\ serge.k.hanna@rutgers.edu, salim.elrouayheb@rutgers.edu
\vspace{-0.1cm}
}
}
\maketitle

\begin{abstract} We consider the problem of constructing codes that can correct $\delta$ deletions occurring in an arbitrary binary string of length $n$ bits. Varshamov-Tenengolts (VT) codes, dating back to 1965,  are zero-error single deletion $(\delta=1)$ correcting codes, and have an asymptotically optimal redundancy. Finding similar codes for $\delta \geq 2$ deletions remains an open problem. In this work, we relax the standard zero-error~(i.e., worst-case) decoding requirement by assuming that the positions of the $\delta$~deletions (or insertions) are independent of the codeword. Our contribution is a new family of explicit codes, that we call Guess~\&~Check (GC) codes, that can correct with high probability up to a constant number of $\delta$ deletions (or insertions). GC codes are systematic; and have deterministic polynomial time encoding and decoding algorithms. We also describe the application of GC codes to file synchronization. 
\end{abstract}

\section{Introduction}
The  deletion channel is probably the most notorious example of a point-to-point channel whose capacity remains unknown. The bits that are deleted by this channel are completely removed from the transmitted sequence and their locations are unknown at the receiver (unlike erasures). For example, if $1010$ is transmitted, the receiver would get $00$ if the first and third bits were deleted. Constructing efficient codes for correcting deletions has also been a challenging task. Varshamov-Tenengolts (VT) codes~\cite{VT65} are the only deletion codes with asymptotically optimal redundancy and can correct only a single deletion. The study of deletion correcting codes has many applications such as file synchronization~\cite{V15,R10,Y14,S16,M11} and DNA-based storage~\cite{RO}. 

In this work, we are particularly motivated by applications of deletion correcting codes to remote file synchronization. Consider the setting where two remote nodes (servers) are connected by a noiseless communication link and have different copies of the same file, where one copy is an edited version of the other. The edits in general include deletions and insertions. We start our discussion by focusing on edits caused by deletions only\footnote{Insertions are discussed in subsequent sections of the paper.}. The goal is to synchronize the two files with minimal number of communicated bits. If a deletion correcting code is {\em systematic}, then it can be used to synchronize the files in the following way. One node uses the code to compute the parity bits of its file and transfers these bits to the other node. Then, the other node uses these parity bits to decode the deletions. In this paper, we construct a new family of explicit codes for multiple deletions, which we call Guess \& Check (GC) codes, that satisfy the systematic property. To the best of our knowledge, the existing multiple deletion correcting codes in the literature (e.g., \cite{B16,D01,R05}) are non-systematic and non-linear\footnote{The non-linearity implies that these codes cannot be made systematic by a linear transformation.}, and hence cannot be applied to remote file synchronization\footnote{For non-systematic codes, the whole file in addition to the redundant bits would need to be transmitted, which is inefficient for remote file synchronization applications.}. 

\subsection{Related work}
The capacity of the deletion channel has been studied in the probabilistic model. In the model where the deletions are iid and occur with a fixed probability $p$, an immediate upper bound on the channel capacity is given by the capacity of the erasure channel $1-p$.  Mitzenmacher and Drinea showed in~\cite{M06} that the capacity of the deletion channel in the iid model is at least $(1-p)/9$. Extensive work in the literature has focused on determining lower and upper bounds on the capacity~\cite{M06,D06,D07,Ven13,K13,Rah15}. We refer interested readers to the  comprehensive survey by Mitzenmacher~\cite{M09}. Ma {\em et al.}~\cite{M11} also studied the capacity in the bursty model of the deletion channel, where the deletion process is modeled by a Markov chain.

A separate line of work has focused on constructing zero-error codes that can correct a given number of deletions $\delta$. Levenshtein  showed in~\cite{L66} that VT codes~\cite{VT65} are capable of correcting a single deletion ($\delta=1$) with zero-error. The redundancy of VT~codes is asymptotically optimal ($\log(n+1)$ bits). More information about the VT codes and other properties of single deletion codes can be found in~\cite{S02}. VT codes have also been used to construct codes that can correct a combination of a single deletion and multiple adjacent transpositions~\cite{RO}. However, finding VT-like codes for multiple deletions ($\delta \geq 2$) is an open problem. In~\cite{L66}, Levenshtein provided bounds showing that the asymptotic number of redundant bits needed to correct $\delta$ bit deletions in an $n$ bit codeword is $c~\delta \log n$ for some constant $c>0$. Levenshtein's bounds were later generalized and improved in~\cite{N14}. The simplest zero-error code for correcting $\delta$~deletions is the $(\delta+1)$ repetition code, where every bit is repeated $(\delta+1)$ times. However, this code is inefficient because it requires $\delta n$ redundant bits, i.e., a redundancy that is linear in $n$. Helberg codes~\cite{H02,A12} are a generalization of VT codes for multiple deletions. These codes can correct multiple deletions but their redundancy is at least linear in $n$ even for two deletions. Some of the works in the literature also studied the problem in the regime where the number of deletions is a constant fraction of $n$~\cite{S99,G14,LG16,Ob} (rather than a constant number). The codes constructed for this regime have $\mathcal{O}(n)$ redundancy. Recently in~\cite{B16}, Brakensiek {\em et al.} constructed zero-error codes with logarithmic redundancy, and polynomial time encoding and decoding, that can decode $\delta$ deletions ($\delta$ fixed) in so-called ``pattern rich" strings, which are a special family of strings that exist for large block lengths (beyond $10^6$ bits for 2 deletions). 
\begin{figure*}
\centering
\resizebox{1\textwidth}{!}{
\begin{tikzpicture}[node distance=2.5cm,auto,>=latex']
\draw (-0.6,0) -- (-0.45,0);
    \node [int] (c) [text width=0.8cm,align=center]{\scriptsize Binary to $q-$ary};
    \node (b) [left of=c,node distance=1.5cm, coordinate] {a};
    \node [int] (z) [right of=c, node distance=3.4cm,text width=2.75cm,align=center] {\scriptsize Systematic MDS $\left(k/\log k+c,k/\log k \right)$};
    \node [int] (y) [right of=z, node distance=3.8cm,text width=0.8cm,align=center] {\scriptsize $q-$ary to binary};
    \node [int] (y1) [right of=y, node distance=3cm, text width=1.5cm,align=center] {\scriptsize {$(\delta+1)$ repetition of parity bits}};
    \node [coordinate] (end) [right of=c, node distance=2cm]{};
    \path[->] (b) edge node {\scriptsize $\mathbf{u}$} node[below] {\scriptsize $k$ bits} (-0.6,0);
    \draw[->] (c) edge node {\scriptsize $\mathbf{U}$} (z) ;
     \node(h1) [right of=c,node distance=1.2cm] [below] {\scriptsize $k/\log k$};
    \node(h2) [below of=h1,node distance=0.3cm] {\scriptsize symbols};
    \node(q1) [below of=c,node distance=0.8cm] {\scriptsize $q=k$};
    \node(q2) [below of=y,node distance=0.8cm] {\scriptsize $q=k$};
    \path[->] (z) edge node {\scriptsize $\mathbf{X}$} (y); 
    \node(i1) [right of=z,node distance=2.37cm] [below] {\scriptsize $k/\log k+c$};
   \node(name)[above of=i1,node distance=1.2cm] [above] {\scriptsize Guess \& Check (GC) codes};
    \node(i2) [below of=i1,node distance=0.3cm] {\scriptsize symbols};
    \node(e1) [right of=y1,node distance=3.3cm] {};
    \path[->] (y1) edge node {\scriptsize $\mathbf{x}$} (e1);
    \path[->] (y) edge node {} (y1);
    \node(t1) [right of=y,node distance=1.33cm] [below] {\scriptsize $k+c\log k$};
    \node(t2) [below of=t1,node distance=0.3cm] {\scriptsize bits};
    \node(tt1) [right of=y1,node distance=2.05cm] [below] {\scriptsize $k+c(\delta +1)\log k$};
    \node(tt2) [below of=tt1,node distance=0.3cm] {\scriptsize bits}; 
    \node(b1) [above of=c,node distance=0.8cm] {\scriptsize Block I};
    \node(b2) [above of=z,node distance=0.565cm] {\scriptsize Block II};
    \node(b3) [above of=y,node distance=0.8cm] {\scriptsize Block III};
    \node(b4) [above of=y1,node distance=0.8cm] {\scriptsize Block IV};
    
    \draw[dashed] (-0.6,-1) rectangle (11.1,1);
\end{tikzpicture} }
\captionsetup{font=footnotesize}
\caption{General encoding block diagram of GC codes for $\delta$ deletions. Block I: The binary message of length $k$ bits is chunked into adjacent blocks of length $\log k$ bits each, and each block is mapped to its corresponding symbol in $GF(q)$ where $q=2^{\log k}=k$. Block II: The resulting string is coded using a systematic $\left(k/\log k+c,k/\log k \right)$ $q-$ary MDS code where $c>\delta$ is the number of parity symbols. Block III: The symbols in $GF(q)$ are mapped to their binary representations. Block IV: Only the parity bits are coded using a $(\delta+1)$ repetition code.}
\label{fig:1}
\end{figure*}
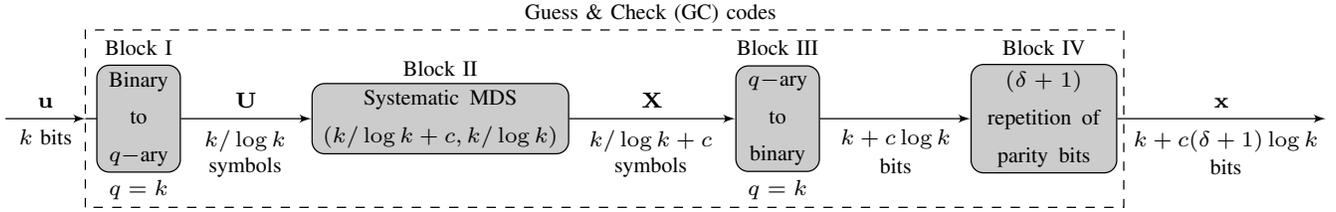
 
 The results mentioned above~\cite{VT65,L66,H02,A12,S99,G14,LG16,B16} focus on correcting deletions with zero-error, which corresponds to correcting worst-case adversarial deletions. There has also been some work on codes that correct random deletions. Davey and Mackay~\cite{D01} constructed watermark codes that can correct iid deletions and insertions. Ratzer~\cite{R05} constructed marker codes and showed that these codes can outperform watermark codes for low deletion/insertion probabilities.
 
 \subsection{Contributions}
In this paper, we construct {\em systematic} codes that can correct multiple deletions with logarithmic redundancy. The systematic property of these codes enables their application to remote file synchronization. Since file synchronization algorithms are probabilistic in general (e.g., \cite{V15,R10,Y14,S16,M11}), we relax the zero-error (worst-case) requirement and construct codes that have a low probability of decoding failure\footnote{A  {\em decoding failure} event corresponds to the decoder not being able to decode the transmitted sequence and outputting ``failure to decode" error message. When the decoder does not output the error message, the decoded string is guaranteed to be the correct one.}. To this end, we assume that the information message is uniform iid, and that the positions of the deletions are independent of the codeword, i.e., oblivious deletions\footnote{Codes that can correct $\delta$ oblivious deletions can also correct $\delta$ deletions that are randomly distributed based on any particular distribution.}. To the best of our knowledge, this is the first work to construct systematic codes for multiple deletions. 

Namely, we make the following  contributions: (i)~We construct new explicit systematic codes, which we call  Guess~\&~Check~(GC) codes,  that can correct, with high probability, and in polynomial time, up to a constant number of deletions $\delta$ (or insertions) occurring in a binary string. The encoding and decoding algorithms of GC codes are deterministic. Moreover, these codes have logarithmic redundancy of value $n-k=c(\delta+1)\log k$, where $k$ and $n$ are the lengths of the message and the codeword, respectively, and $c>\delta$ is a code parameter that is independent of $k$; (ii)~GC codes enable different trade-offs between redundancy, decoding complexity, and probability of decoding failure; (iii)~We implemented GC codes and the programming code can be found and tested online on the link in~\cite{Software}. Based on our implementations, we provide numerical simulations on the decoding failure of GC codes and compare these simulations to our theoretical results. For instance, we observe that a GC code with rate 0.8 can correct up to 4~deletions in a message of length 1024~bits with no decoding failure detected within 10000 runs of simulations;
(iv)~We demonstrate how GC codes can be used as a building block for file synchronization algorithms by including these codes as part of the synchronization algorithm proposed by Venkataramanan {\em et al.} in \cite{V15,R10}. As a result, we provide numerical simulations highlighting the savings in number of rounds and total communication cost. 

\subsection{Organization} 
The paper is organized as follows. In Section~\ref{sec:2}, we introduce the necessary notations used throughout the paper. We state and discuss the main result of this paper in Section~\ref{sec:3}. In Section~\ref{sec:4}, we provide encoding and decoding examples on GC codes. In Section~\ref{sec:5}, we describe in detail the encoding and decoding schemes of GC codes.
The proof of the main result of this paper is given in Section~\ref{sec:6}. In Section~\ref{sec:6.1}, we explain the trade-offs achieved by GC codes. In Section~\ref{sec:i}, we explain how these codes can be used to correct $\delta$ insertions instead of $\delta$ deletions. The results of the numerical simulations on the decoding failure of GC codes and their application to file synchronization are shown in Sections~\ref{sec:7} and~\ref{sec:8}, respectively. We conclude with some open problems in Section~\ref{sec:11}.
\section{Notation}
\label{sec:2}
Let $k$ and $n$ be the lengths in bits of the message and codeword, respectively.  Let $\delta$ be a constant representing the number of deletions. Without loss of generality, we assume that $k$ is a power of $2$. GC codes are based on $q-$ary systematic \mbox{$\left(\left \lceil k/\log k \right \rceil + c,\left \lceil k/\log k \right \rceil \right)$} MDS codes, where \mbox{$q=k>\left \lceil k/\log k \right \rceil + c$} and \mbox{$c>\delta$} is a code parameter representing the number of MDS parity symbols\footnote{Explicit constructions of systematic Reed-Solomon codes (based on Cauchy or Vandermonde matrices) always exist for these parameters.}.
We will drop the ceiling notation for $\left \lceil k/\log k \right \rceil$ and simply write $k/\log k$. All logarithms in this paper are of base $2$. The block diagram of the encoder is shown in Fig.~\ref{fig:1}. We denote binary and $q-$ary  vectors by lower and upper case bold letters respectively, and random variables by calligraphic letters.

\section{Main Result}
\label{sec:3}
Let $\mathbf{u}$ be a binary vector of length $k$ with iid Bernoulli$(1/2)$ components representing the information message. The message $\mathbf{u}$ is encoded into the codeword $\mathbf{x}$ of length $n$ bits using the Guess \& Check (GC) code illustrated in Fig.~\ref{fig:1}. The GC decoder, explained in Section~\ref{sec:5}, can either decode successfully and output the decoded string, or output a ``failure to decode" error message because it cannot make a correct decision. The latter case is referred to as a {\em decoding failure}, and its corresponding event is denoted by $F$.
\begin{theorem} 
\label{thm:1} 
Guess \& Check (GC) codes can correct in polynomial time up to a constant number of $\delta$ deletions occurring within~$\mathbf{x}$. Let $c>\delta$ be a constant integer. The code has the following properties:
\begin{enumerate}[leftmargin=*]
\item Redundancy: $n-k = c(\delta+1) \log k$ bits. 
\item Encoding complexity is $\mathcal{O}(k \log k)$, and  decoding complexity is $\mathcal{O}\left(\frac{k^{\delta+1}}{\log^{\delta-1}k}\right)$.
\item For any $\delta$ deletion positions chosen independently of $\mathbf{x}$, the probability that the decoding fails for a uniform iid message is: $Pr(F) = \mathcal{O} \left( \frac{k^{2\delta-c}}{\log^{\delta} k} \right)$. 
\end{enumerate}
\end{theorem}

\noindent The result on the probability of decoding failure in Theorem~\ref{thm:1} applies for any $\delta$ deletion positions which are chosen independently of the codeword~$\mathbf{x}$. Hence, the same result can be also obtained for any random distribution over the $\delta$ deletion positions (like the uniform distribution for example), by averaging over all the possible positions of the $\delta$ deletions.

GC codes enable trade-offs between the code properties shown in Theorem~\ref{thm:1}, this will be highlighted later in Section~\ref{sec:6.1}. The code properties show the following:~(i)~the redundancy is logarithmic in $k$, and the code rate, \mbox{$R=k/(k+c(\delta+1) \log k)$}, is asymptotically optimal in $k$ (approaches $1$ as $k$ goes to infinity); (ii)~the order of complexity is polynomial in $k$ and is not affected by the constant $c$; (iii)~the probability of decoding failure goes to zero polynomially in $k$ if $c>2\delta$; and exponentially in $c$ for a fixed $k$. Note that the decoder can always detect when it cannot decode successfully. This can serve as an advantage in models which allow feedback. There, the decoder can ask for additional redundancy to be able to decode successfully.
\section{Examples} \label{sec:4} Guess \& Check (GC) codes can correct up to $\delta$ deletions with high probability. We provide examples  to illustrate the encoding and decoding schemes. The examples are for $\delta=1$ deletion just for the sake of simplicity\footnote{VT codes can correct one deletion with zero-error. However, GC codes are generalizable to multiple deletions.}.  
\begin{example}[Encoding]
\label{ex:1}
Consider a binary message $\mathbf{u}$  of length $k=16$ given by
\begin{equation*}
\mathbf{u}=1~1~1~0~0~0~0~0~1~1~0~1~0~0~0~1.
\end{equation*}
$\mathbf{u}$ is encoded by following the different encoding blocks illustrated in Fig.~\ref{fig:1}. \\
\noindent $1)$ {\em Binary to $q-$ary (Block I, Fig.~\ref{fig:1})}. The message $\mathbf{u}$ is chunked into adjacent blocks of length $\log k=4$ bits each,
\begin{equation*}
\mathbf{u}=\underbrace{\overbracket{1~1~1~0}^{\text{block 1}}}_{\alpha^{11}}~\underbrace{\overbracket{0~0~0~0}^{\text{block 2}}}_{0}~\underbrace{\overbracket{1~1~0~1}^{\text{block 3}}}_{\alpha^{13}}~\underbrace{\overbracket{0~0~0~1}^{\text{block 4}}}_{1}\color{black}.
\end{equation*}
Each block is then mapped to its corresponding symbol in $GF(q)$, $q=k=2^4=16$, by considering its leftmost bit as its most significant bit. This results in a string $\mathbf{U}$ which consists of $k/\log k=4$ symbols in $GF(16)$. The extension field used here has a primitive element $\alpha$, with $\alpha^4=\alpha+1$. Hence, $\mathbf{U}\in GF(16)^4$ is given by
\begin{equation*}
\mathbf{U}=(\alpha^{11},0,\alpha^{13},1\color{black}).
\end{equation*}
$2)$ {\em Systematic MDS code (Block II, Fig.~\ref{fig:1})}. $\mathbf{U}$ is then coded using a systematic $(k/\log k +c,k/\log k)=(6,4)$ MDS code over $GF(16)$, with $c=2>\delta$. The encoded string is denoted by $\mathbf{X} \in GF(16)^6$ and is given by multiplying $\mathbf{U}$ by the following code generator matrix,
\begin{align*}
\mathbf{X} &= \left(\alpha^{11},0,\alpha^{13},1\color{black}\right)
\left(
\begin{matrix} \color{blue}
\color{blue} 1 & \color{blue}0 & \color{blue}0 & \color{blue}0 & \color{red}1 & \color{red}1\\
\color{blue}0 & \color{blue}1 & \color{blue}0 & \color{blue}0 & \color{red}1 & \color{red}\alpha\\
\color{blue}0 & \color{blue}0 & \color{blue}1 & \color{blue}0 & \color{red}1 & \color{red}\alpha^2\\
\color{blue}0 & \color{blue}0 & \color{blue}0 & \color{blue}1 & \color{red}1 & \color{red}\alpha^3 \\
\end{matrix} \right), \\ 
&= (\color{blue} \alpha^{11},0,\alpha^{13},1,\color{red} \alpha,\alpha^{10} \color{black}).
\end{align*}
$3)$ {\em $Q-$ary to binary (Block III, Fig.~\ref{fig:1})}. The binary codeword corresponding to $\mathbf{X}$, of length $n=k+2\log k=24$ bits, is
\begin{equation*}
\mathbf{x}=\color{blue} \underbrace{1~1~1~0}~\underbrace{0~0~0~0}~\underbrace{1~1~0~1}~\underbrace{0~0~0~1}~\color{red}\underbrace{\color{red}0~0~1~0}~\color{red}\underbrace{\color{red} 0~1~1~1 }\color{black}.
\end{equation*}
For simplicity we skip the last encoding step (Block IV) intended to protect the parity bits and assume  that deletions affect only the systematic bits.
\end{example}
The high level idea of the decoding algorithm is to: (i) make an assumption on in which block the bit deletion has occurred (the guessing part); (ii) chunk the bits accordingly, treat the affected block as erased, decode the erasure and check whether the obtained sequence is consistent with the parities (the checking part); (iii) go over  all the possibilities.
\begin{example}[Successful Decoding] \label{ex:1a}
Suppose that the $14^{th}$ bit of $\mathbf{x}$ gets deleted,
\begin{equation*}
\mathbf{x}=\color{blue}1~1~1~0~0~0~0~0~1~1~0~1~0~\textbf{\underline{0}}~0~1~\color{red}0~0~1~0~0~1~1~1\color{black}.
\end{equation*}
The decoder receives the following $23$ bit string $\mathbf{y}$,
\begin{equation*}
\mathbf{y}=\color{blue}1~1~1~0~0~0~0~0~1~1~0~1~0~0~1~\color{red}0~0~1~0~0~1~1~1\color{black}.
\end{equation*}
The decoder goes through all the possible $k/\log k=4$ cases, where in each case $i$, $i=1,\ldots,4$, the deletion is assumed to have occurred in block $i$ and $\mathbf{y}$ is chunked accordingly. Given this assumption, symbol $i$ is considered erased and erasure decoding is applied over $GF(16)$ to recover this symbol. Furthermore, given two parities, each symbol $i$ can be recovered in two different ways. Without loss of generality, we assume that the first parity $p_1$, $p_1=\alpha$, is the parity used for decoding the erasure. The decoded $q-$ary string in case $i$ is denoted by $\mathbf{Y_i}\in GF(16)^4$, and its binary representation is denoted by $\mathbf{y_i}\in GF(2)^{16}$. The four cases are shown below: \\
{\em \underline{Case 1}:} The deletion is assumed to have occurred in block 1, so $\mathbf{y}$ is chunked as follows and the erasure is denoted by  $\mathcal{E}$,
\begin{equation*}
\color{blue} \underbrace{1~1~1}_{\mathcal{E}}~\underbrace{0~0~0~0}_{0}~\underbrace{0~1~1~0}_{\alpha^{5}}~\underbrace{1~0~0~1}_{\alpha^{14}}~\color{red}\underbrace{0~0~1~0}_{\alpha}~\underbrace{0~1~1~1}_{\alpha^{10}}\color{black}.
\end{equation*}
Applying erasure decoding over $GF(16)$, the recovered value of symbol 1 is $\alpha^{13}$. Hence, the decoded $q-$ary string \mbox{$\mathbf{Y_1} \in GF(16)^4$} is
\begin{equation*}
\mathbf{Y_1}=(\color{blue}\alpha^{13},0,\alpha^{5},\alpha^{14}\color{black}) .
\end{equation*}
Its equivalent in binary $\mathbf{y_1}\in GF(2)^{16}$ is
\begin{equation*}
\mathbf{y_1} = \color{blue}\underbrace{1~1~0~1}_{\alpha^{13}}~\underbrace{0~0~0~0}_{0}~\underbrace{0~1~1~0}_{\alpha^{5}}~\underbrace{1~0~0~1}_{\alpha^{14}}\color{black}.
\end{equation*}
Now, to check  our assumption, we test whether $\mathbf{Y_1}$ is consistent with the second parity $p_2=\alpha^{10}$. However, the computed parity is 
\begin{equation*}
\left(\color{blue}\alpha^{13},0,\alpha^{5},\alpha^{14}\color{black}\right)\left(1,\alpha,\alpha^{2},\alpha^{3}\right)^{T}=\alpha \neq \color{red} \alpha^{10}.
\end{equation*}
This shows that $\mathbf{Y_1}$ does not satisfy the second parity. Therefore, we deduce that our assumption on the deletion location  is wrong, i.e., the deletion did not occur in block $1$. Throughout the paper we refer to such cases as {\em impossible} cases.\\
{\em \underline{Case 2}:}  The deletion is assumed to have occurred in block 2, so the sequence is chunked as follows
\begin{equation*}
\color{blue} \underbrace{1~1~1~0}_{\alpha^{11}}~\underbrace{0~0~0}_{\mathcal{E}}~\underbrace{0~1~1~0}_{\alpha^{5}}~\underbrace{1~0~0~1}_{\alpha^{14}}~\color{red}\underbrace{0~0~1~0}_{\alpha}~\underbrace{0~1~1~1}_{\alpha^{10}}\color{black}.
\end{equation*}
Applying erasure decoding, the recovered value of symbol 2 is $\alpha^4$. Now, before checking whether the decoded string is consistent with the second parity $p_2$, one can notice that the binary representation of the decoded erasure $(0011)$ is not a supersequence of the sub-block $(000)$. So, without checking $p_2$, we can deduce that this case is {\em impossible}. 
\begin{definition} \label{def:1} We restrict this definition to the case of $\delta=1$ deletion with two MDS parity symbols in $GF(q)$. A case $i$, $i=1,2,\ldots,k/\log k$, is said to be {\em possible} if it satisfies the two criteria below simultaneously. \\
Criterion 1: The $q-$ary string decoded based on the first parity in case $i$, denoted by $\mathbf{Y_i}$, satisfies the second parity. \\
Criterion 2: The binary representation of the decoded erasure is a supersequence of its corresponding sub-block. \\
If any of the two criteria is not satisfied, the case is said to be {\em impossible}.
\end{definition}
\noindent The two criteria mentioned above are both necessary. For instance, in this example, case 2 does not satisfy Criterion 2 but it is easy to verify that it satisfies Criterion 1. Furthermore, case 1 satisfies Criterion 1 but does not satisfy Criterion 2. A case is said to be {\em possible} if it satisfies both criteria simultaneously. \\
\noindent {\em \underline{Case 3}:} The deletion is assumed to have occurred in block 3, so the sequence is chunked as follows
\begin{equation*}
\color{blue} \underbrace{1~1~1~0}_{\alpha^{11}}~\underbrace{0~0~0~0}_{0}~\underbrace{1~1~0}_{\mathcal{E}}~\underbrace{1~0~0~1}_{\alpha^{14}}~\color{red}\underbrace{0~0~1~0}_{\alpha}~\underbrace{0~1~1~1}_{\alpha^{10}}\color{black}.
\end{equation*}
In this case, the decoded binary string is
\begin{equation*}
\mathbf{y_3} = \color{blue}\underbrace{1~1~1~0}_{\alpha^{11}}~\underbrace{0~0~0~0}_{0}~\underbrace{0~1~0~1}_{\alpha^{8}}~\underbrace{1~0~0~1}_{\alpha^{14}}~\color{black}.
\end{equation*}
By following the same steps as cases 1 and 2, it is easy to verify that both criteria are not satisfied in this case, i.e., case 3 is also {\em impossible}. \\
\noindent {\em \underline{Case 4}:} The deletion is assumed to have occurred in block 4, so the sequence is chunked as follows
\begin{equation*}
\color{blue} \underbrace{1~1~1~0}_{\alpha^{11}}~\underbrace{0~0~0~0}_{0}~\underbrace{1~1~0~1}_{\alpha^{13}}~\underbrace{0~0~1}_{\mathcal{E}}~\color{red}\underbrace{0~0~1~0}_{\alpha}~\underbrace{0~1~1~1}_{\alpha^{10}}\color{black}.
\end{equation*}
In this case, the decoded binary string is 
\begin{equation*}
\mathbf{y_4} = \color{blue}\underbrace{1~1~1~0}_{\alpha^{11}}~\underbrace{0~0~0~0}_{0}~\underbrace{1~1~0~1}_{\alpha^{13}}~\underbrace{0~0~0~1}_{1}~\color{black}.
\end{equation*}
Here, it is easy to verify that this case satisfies both criteria and is indeed {\em possible}. \\ 
After going through all the cases, case 4 stands alone as the only {\em possible} case. So the decoder declares successful decoding and outputs $\mathbf{y_4}$ ($\mathbf{y_4}=\mathbf{u}$). 
\end{example}

\noindent The next example considers another message $\mathbf{u}$ and shows how the proposed decoding scheme can lead to a decoding failure. The importance of Theorem~\ref{thm:1} is that it shows that the probability of a decoding failure vanishes a $k$ goes to infinity.

\begin{example}[Decoding failure]
\label{ex:2}
Let $k=16, \delta=1$ and $c=2$. Consider the binary message $\mathbf{u}$ given by
\begin{equation*}
\mathbf{u}=1~1~0~1~0~0~0~0~1~0~0~0~0~1~0~1\color{black}.
\end{equation*}
Following the same encoding steps as before, the $q-$ary codeword $\mathbf{X}\in GF(16)^6$ is given by
\begin{equation*}
\mathbf{X}=(\color{blue}\alpha^{13},0,\alpha^{3},\alpha^{8},\color{red}0,\alpha^{8}\color{black}).
\end{equation*}
We still assume that the deletion affects only the systematic bits. Suppose that the $14^{th}$ bit of the binary codeword \mbox{$\mathbf{x}\in GF(2)^{24}$} gets deleted
\begin{equation*}
\mathbf{x}=\color{blue} 1~1~0~1~0~0~0~0~1~0~0~0~0~\textbf{\underline{1}}~0~1~\color{red}0~0~0~0~0~1~0~1 \color{black}.
\end{equation*}
The decoder receives the following $23$ bit binary string \mbox{$\mathbf{y}\in GF(2)^{23}$},
\begin{equation*}
\mathbf{y}=\color{blue} 1~1~0~1~0~0~0~0~1~0~0~0~0~0~1~\color{red}0~0~0~0~0~1~0~1 \color{black}.
\end{equation*}
The decoding is carried out as explained in Example~\ref{ex:1a}. The $q-$ary strings decoded in case $1$ and case $4$ are
\begin{align*}
\mathbf{Y_1}&=(\color{blue}\alpha^{13},\alpha^{3},\alpha^{2},1\color{black}), \\
\mathbf{Y_4}&=(\color{blue}\alpha^{13},0,\alpha^{3},\alpha^{8}\color{black}).
\end{align*}
It is easy to verify that both cases $1$ and $4$ are {\em possible} cases. The decoder here cannot know which of the two cases is the correct one, so it declares a decoding failure.
\end{example}
\begin{remark}
In the previous analysis, each case refers to the assumption that a certain block is affected by the deletion. Hence, among all the cases considered, there is only one correct case that corresponds to the actual deletion location. That correct case always satisfies the two criteria for {\em possible} cases (Definition~\ref{def:1}). So whenever there is only one {\em possible} case (like in Example~\ref{ex:1a}), the decoding will be successful since that case would be for sure the correct one. However, in general, the analysis may yield multiple {\em possible} cases. Nevertheless, the decoding can still be successful if all these {\em possible} cases lead to the same decoded string. An example of this is when the transmitted codeword is the all $0$'s sequence. Regardless of the deletion position, this sequence will be decoded as all $0$'s in all the cases. In fact, whenever the deletion occurs within a run of $0$'s or $1$'s that extends to multiple blocks, the cases corresponding to these blocks will all be {\em possible} and lead to the same decoded string. However, sometimes the {\em possible} cases can lead to different decoded strings like in Example~\ref{ex:2}, thereby causing a decoding failure. 
\end{remark}
\section{General Encoding and Decoding of GC Codes}
\label{sec:5}
In this section, we describe the general encoding and decoding schemes that can correct up to $\delta$ deletions. The encoding and decoding steps for $\delta>1$ deletions are a direct generalization of the steps for $\delta=1$ described in the previous section. For decoding, we assume without loss of generality that exactly $\delta$ deletions have occurred. Therefore, the length of the binary string $\mathbf{y}$ received by the decoder is $n-\delta$ bits.
\subsection{Encoding Steps}
\noindent The encoding steps follow from the block diagram that was shown in Fig.~\ref{fig:1}. \\
\noindent $1)$ {\em Binary to $q-$ary (Block I, Fig.~\ref{fig:1})}. The message $\mathbf{u}$ is chunked into adjacent blocks of length $\log k$ bits each. Each block is then mapped to its corresponding symbol in $GF(q)$, where $q=2^{\log k}=k$. This results in a string $\mathbf{U}$ which consists of $k/\log k$ symbols in $GF(q)$ \footnote{After chunking, the last block may contain fewer than $\log k$ bits. In order to map the block to its corresponding symbol, it is first padded with zeros to a length of $\log k$ bits. Hence, $\mathbf{U}$ consists of $\left \lceil k/\log k \right \rceil$ symbols. We drop the ceiling notation throughout the paper and simply write $k/\log k$.}.  \\
\noindent $2)$ {\em Systematic MDS code (Block II, Fig.~\ref{fig:1})}. $\mathbf{U}$ is then coded using a systematic $(k/\log k + c, k/\log k)$ MDS code over $GF(q)$, where $c>\delta$ is a code parameter. The $q-$ary codeword $\mathbf{X}$ consists of $k/\log k+c$ symbols in $GF(q)$. \\
\noindent $3)$ {\em $Q-$ary to binary (Block III, Fig.~\ref{fig:1})}. The binary representations of the symbols in $\mathbf{X}$ are concatenated respectively. \\
\noindent $4)$ {\em Coding parity bits by repetition (Block IV, Fig.~\ref{fig:1})}. Only the parity bits are coded using a $(\delta+1)$ repetition code, i.e., each bit is repeated $(\delta+1)$ times. The resulting binary codeword $\mathbf{x}$ to be transmitted is of length $n= k+c\left(\delta +1\right)\log k$.
\subsection{Decoding Steps}
\noindent $1)$ Decoding the parity symbols of Block II (Fig.~\ref{fig:1}): these parities are protected by a $(\delta+1)$ repetition code, and therefore can be always recovered correctly by the decoder. A simple way to do this is to examine the bits of $\mathbf{y}$ from right to left and decode deletions instantaneously until the total length of the decoded sequence is $c(\delta +1)$ bits (the original length of the coded parity bits). Therefore, for the remaining steps we will  assume without loss of generality that all the $\delta$ deletions have occurred in the systematic bits.\\
\noindent $2)$ The guessing part: the number of possible ways to distribute the $\delta$ deletions among the $k/\log k$ blocks is
\begin{equation*}
t \triangleq \binom{k/\log k+\delta-1}{\delta}.
\end{equation*}
We index these possibilities by $i,i=1,\ldots,t,$ and refer to each possibility by case $i$. \\
\noindent The decoder goes through all the $t$ cases (guesses). \\
\noindent $3)$ The checking part: for each case $i$, $i=1,\ldots,t$, the decoder (i) chunks the sequence according to the corresponding assumption; (ii) considers the affected blocks erased and maps the remaining blocks to their corresponding symbols in $GF(q)$; (iii) decodes the erasures using the first $\delta$ parity symbols; (iv)~checks whether the case is {\em possible} or not based on the criteria described below.
\begin{definition} \label{def:4} For $\delta$ deletions, a case $i$, $i=1,2,\ldots,t$, is said to be {\em possible} if it satisfies the two criteria below simultaneously. \\
Criterion 1: the decoded $q-$ary string in case $i$, denoted by $\mathbf{Y_i} \in GF(q)^{k/\log k}$, satisfies the last $c-\delta$ parities simultaneously. \\
Criterion 2: The binary representations of all the decoded erasures in $\mathbf{Y_i}$ are supersequences of their corresponding sub-blocks. \\
If any of these two criteria is not satisfied, the case is said to be {\em impossible}.
\end{definition}
\noindent $4)$ After going through all the cases, the decoder declares successful decoding if (i) only one {\em possible} case exists; or (ii)~multiple {\em possible} cases exist but all lead to the same decoded string. Otherwise, the decoder declares a decoding failure.
\section{Proof of Theorem~\ref{thm:1}}
\label{sec:6}
\subsection{Redundancy}
The $(k/\log k +c, k/\log k)$ $q-$ary MDS code in step~$2$ of the encoding scheme adds a redundancy of $c\log k$ bits. These $c\log k$ bits are then coded using a $(\delta+1)$ repetition code. Therefore, the overall redundancy of the code is $c(\delta+1)\log k$ bits. 
\subsection{Complexity} 
{\em i) Encoding Complexity}: The complexity of mapping a binary string to its corresponding $q-$ary symbol (step~$1$), or vice versa (step~$3$), is $\mathcal{O}(k)$. Furthermore, the encoding complexity of a $(k/\log k +c, k/\log k)$ $q-$ary systematic MDS code is quantified by the complexity of computing the $c$ MDS parity symbols. Computing one MDS parity symbol involves $k/\log k$ multiplications of symbols in $GF(q)$. The complexity of multiplying two symbols in $GF(q)$ is $\mathcal{O}(\log^2 q)$. Recall that in our code $q=k$. Therefore, the complexity of step~$2$ is $\mathcal{O}(\log^2 k \cdot c \cdot k/\log k)= \mathcal{O}(c \cdot k\log k)$. Step~$4$ in the encoding scheme codes $c\log k$ bits by repetition, its complexity is $\mathcal{O}(c\log k)$. Therefore, the encoding complexity of GC codes is $\mathcal{O}(c \cdot k\log k)=\mathcal{O}(k\log k)$ since $c>\delta$ is a constant.

{\em ii) Decoding Complexity}: The computationally dominant step in the decoding algorithm is step~$3$, that goes over all the $t$ cases and decodes the erasures in each case\footnote{The complexity of checking whether a decoded erasure, of length $\log k$ bits, is a supersequence of its corresponding sub-block is $\mathcal{O}(\log^2 k)$ using the Wagner-Fischer algorithm. Hence, Criterion 2 (Definition~\ref{def:4}) does not affect the order of decoding complexity.}. Since the erasures to be decoded are within the systematic bits, then decoding $\delta$ erasures in each case can be done by performing the following steps: (a)~multiplying the unerased systematic symbols by the corresponding $\delta$ MDS encoding vectors; (b)~subtracting the obtained results from the values of the corresponding MDS parities; (c)~inverting a $\delta \times \delta$ matrix. Since $\delta$ is a constant, the complexity of the previous steps is $\mathcal{O}(\log^2 q\cdot k/\log k)=\mathcal{O}(k\log k)$. Now since the erasure decoding is performed for all the $t$ cases, the total decoding complexity is $\mathcal{O}\left(t \cdot k\log k\right)$.
The number of cases $t$ is given by
\begin{equation}
\binom{k/\log k+\delta-1}{\delta} = \mathcal{O} \left(\frac{k^{\delta}}{\log^{\delta} k}\right). \label{bin}
\end{equation}
\eqref{bin} follows from the fact that $\binom{b}{a}\leq b^a$ for all integer values of $a$ and $b$ such that $1\leq a \leq b$. Therefore, the overall decoding complexity is
\begin{equation*}
\mathcal{O}\left(\frac{k^{\delta+1}}{\log^{\delta-1} k}\right),
\end{equation*}
which is polynomial in $k$ for constant $\delta$. 
\subsection{Proof of the probability of decoding failure for $\delta=1$}
To prove the upper bound on the probability of decoding failure, we first introduce the steps of the proof for $\delta=1$ deletion. Then, we generalize the proof to the case of $\delta>1$. 

\vspace{0.1cm} 

The probability of decoding failure for $\delta=1$ is computed over all possible $k-$bit messages. Recall that the bits of the message $\mathbf{u}$ are iid Bernoulli$(1/2)$. The message $\mathbf{u}$ is encoded as shown in Fig.~\ref{fig:1}. For $\delta=1$, the decoder goes through a total of $k/\log k$ cases, where in a case $i$ it decodes by assuming that block $i$ is affected by the deletion. Let $\bm{\mathcal{Y}_i}$ be the random variable representing the $q-$ary string decoded in case \mbox{$i$, $i=1,2,\ldots,k/\log k$}, in step 3 of the decoding scheme. Let $\mathbf{Y} \in GF(q)^{k/\log k}$ be a realization of the random variable $\bm{\mathcal{Y}_i}$.
We denote by $\mathcal{P}_r \in GF(q), r=1,2,\ldots,c,$ the random variable representing the $r^{th}$ MDS parity symbol (Block II, Fig.~\ref{fig:1}). Also, let $\mathbf{G_r} \in GF(q)^{k/\log k}$ be the MDS encoding vector responsible for generating $\mathcal{P}_r$. Consider $c>\delta$ arbitrary MDS parities $p_1,\ldots,p_c$, for which we define the following sets. For $r=1,\ldots,c,$
\begin{align*}
\mathrm{A_r} &\triangleq \{ \mathbf{Y} \in GF(q)^{k/\log k} |~\mathbf{G_r^TY}=p_r\}, \\
\mathrm{A} &\triangleq \mathrm{A_1} \cap \mathrm{A_2} \cap \ldots \cap \mathrm{A_c}.
\end{align*}
$\mathrm{A_r}$ and $\mathrm{A}$ are affine subspaces of dimensions $k/\log k -1$ and $k/\log k-c$, respectively. Therefore,
\begin{equation}
\label{e:A}
\left \lvert \mathrm{A_r} \right \rvert = q^{\frac{k}{\log k}-1} \text{ and} ~ \left \lvert \mathrm{A} \right \rvert = q^{\frac{k}{\log k}-c}. 
\end{equation}
Recall that the correct values of the MDS parities are recovered at the decoder, and that for $\delta=1$, $\bm{\mathcal{Y}_i}$ is decoded based on the first parity. Hence, for a fixed MDS parity $p_1$, and for $\delta=1$ deletion, $\bm{\mathcal{Y}_i}$ takes values in $\mathrm{A_1}$. Note that $\bm{\mathcal{Y}_i}$ is not necessarily uniformly distributed over $\mathrm{A_1}$. For instance, if the assumption in case $i$ is wrong, two different message inputs can generate the same decoded string $\bm{\mathcal{Y}_i} \in \mathrm{A_1}$.  We illustrate this later through Example~\ref{ex:3}. The next claim gives an upper bound on the probability mass function of  $\bm{\mathcal{Y}_i}$ for $\delta=1$ deletion. 
\begin{claim} 
\label{claim:1}
For any case $i$, $i=1,2,\ldots,k/\log k$,
\begin{equation*}
Pr\left(\bm{\mathcal{Y}_i}=\mathbf{Y} | \mathcal{P}_1=p_1 \right) \leq \frac{2}{q^{\frac{k}{\log k}-1}}.
\end{equation*}
\end{claim}
\noindent  Claim~\ref{claim:1} can be interpreted as that at most $2$ different input messages can generate the same decoded string $\bm{\mathcal{Y}_i} \in \mathrm{A_1}$. We assume Claim~\ref{claim:1} is true for now and prove it later in Section~\ref{clm1}. Next, we use this claim to prove the following upper bound on the probability of decoding failure for $\delta=1$,
\begin{equation}
Pr(F) < \frac{2}{k^{c-2}~\log k}. \label{d:1}
\end{equation}
\noindent In the general decoding scheme, we mentioned two criteria which determine whether a case is {\em possible} or not (Definition~\ref{def:4}). Here, we upper bound $Pr(F)$ by taking into account Criterion~$1$ only. Based on Criterion~$1$, if a case $i$ is possible, then $\bm{\mathcal{Y}_i}$ satisfies all the $c$ MDS parities simultaneously, i.e., $\bm{\mathcal{Y}_i} \in \mathrm{A}$. Without loss of generality, we assume case~$1$ is the correct case, i.e., the deletion occurred in block~$1$. A decoding failure is declared if there exists a {\em possible} case~\mbox{$j$, $j=2,\ldots,k/\log k$}, that leads to a decoded string different than that of case~$1$. Namely, $\bm{\mathcal{Y}_j} \in \mathrm{A}$ and $\bm{\mathcal{Y}_j} \neq \bm{\mathcal{Y}_1}$. Therefore,
\begin{align}
Pr\left(F | \mathcal{P}_1=p_1 \right) &\leq Pr\left(\bigcup_{j=2}^{k/\log k}\{\bm{\mathcal{Y}_j}\in \mathrm{A},\bm{\mathcal{Y}_j} \neq \bm{\mathcal{Y}_1}\} \biggr\rvert \mathcal{P}_1=p_1 \right) \\
&\leq \sum_{j=2}^{k/\log k} Pr\left(\bm{\mathcal{Y}_j}\in \mathrm{A},\bm{\mathcal{Y}_j} \neq \bm{\mathcal{Y}_1} | \mathcal{P}_1=p_1 \right) \label{e:2} \\
&\leq \sum_{j=2}^{k/\log k} Pr\left(\bm{\mathcal{Y}_j}\in \mathrm{A} | \mathcal{P}_1=p_1 \right) \label{e:4}\\
&= \sum_{j=2}^{k/\log k} \sum_{\mathbf{Y}\in \mathrm{A}} Pr\left(\bm{\mathcal{Y}_j}=\mathbf{Y} | \mathcal{P}_1=p_1 \right) \\
&\leq \sum_{j=2}^{k/\log k}  \sum_{\mathbf{Y}\in \mathrm{A}} \frac{2}{q^{\frac{k}{\log k}-1}} \label{e:6}\\
&= \sum_{j=2}^{k/\log k}  \left \lvert \mathrm{A} \right \rvert \frac{2}{q^{\frac{k}{\log k}-1}}  \\
&= \sum_{j=2}^{k/\log k}  2~\frac{q^{\frac{k}{\log k} -c}}{q^{\frac{k}{\log k}-1}} \label{e:8} \\
&=  \left(\frac{k}{\log k}-1\right)\frac{2}{q^{c-1}} \\
&< \frac{2}{k^{c-2}~\log k} \label{e:10}.
\end{align}
\noindent \eqref{e:2} follows from applying the union bound. \eqref{e:4} follows from the fact that $Pr\left(\bm{\mathcal{Y}_j} \neq \bm{\mathcal{Y}_1}|\bm{\mathcal{Y}_j}\in \mathrm{A}, \mathcal{P}_1=p_1 \right)\leq 1$. \eqref{e:6} follows from Claim~\ref{claim:1}. \eqref{e:8} follows from \eqref{e:A}. \eqref{e:10} follows from the fact that $q=k$ in the coding scheme. Next, to complete the proof of \eqref{d:1}, we use \eqref{e:10} and average over all values of $p_1$.
\begin{align*}
Pr(F) &= \sum_{p_1\in GF(q)} Pr\left(F|\mathcal{P}_1=p_1\right)Pr\left(\mathcal{P}_1=p_1\right)  \\
&< \sum_{p_1\in GF(q)} \frac{2}{k^{c-2}~\log k}~Pr\left(\mathcal{P}_1=p_1\right) \\
&= \frac{2}{k^{c-2}~\log k}.
\end{align*}
\subsection{Proof of the probability of decoding failure for $\delta$ deletions}
We now generalize the previous proof for $\delta>1$ deletions and show that
\begin{equation*}
Pr(F) = \mathcal{O} \left( \frac{1}{k^{c-2\delta}~\log^{\delta} k} \right).
\end{equation*}
For $\delta$ deletions the number of cases is given by
\begin{equation}
\label{ttt}
t=\binom{k/\log k+\delta-1}{\delta} = \mathcal{O} \left(\frac{k^{\delta}}{\log^{\delta} k}\right).
\end{equation}
Consider the random variable $\bm{\mathcal{Y}_i}$ which represents the $q-$ary string decoded in case $i$, $i=1,2,\ldots,t$. The next claim generalizes Claim~\ref{claim:1} for $\delta>1$ deletions. 
\begin{claim} 
\label{claim:2}
There exists a deterministic function $h$ of $\delta$, $h(\delta)$ independent of $k$, such that for any case $i$, $i=1,2,\ldots,t$,
\begin{equation*}
Pr\left(\bm{\mathcal{Y}_i}=\mathbf{Y} | \mathcal{P}_1=p_1,\ldots,\mathcal{P}_{\delta}=p_{\delta}\right) \leq \frac{h(\delta)}{q^{\frac{k}{\log k}-\delta}}.
\end{equation*}
\end{claim}
\noindent We assume Claim~\ref{claim:2} is true for now and prove it later (see Appendix~\ref{app:A}). To bound the probability of decoding failure for $\delta$ deletions, we use the result of Claim~\ref{claim:2} and follow the same steps of the proof for $\delta=1$ deletion while considering $t$ cases instead of $k/\log k$. Some of the steps will be skipped for the sake of brevity. 
\begin{align}
Pr\left(F|p_1,\ldots,p_{\delta}\right) &\leq Pr\left(\bigcup_{j=2}^{t}\{\bm{\mathcal{Y}_j}\in \mathrm{A},\bm{\mathcal{Y}_j} \neq \bm{\mathcal{Y}_1}\} \biggr\rvert p_1,\ldots,p_{\delta} \right) \\
&\leq \sum_{j=2}^{t} Pr\left(\bm{\mathcal{Y}_j}\in \mathrm{A} | p_1,\ldots,p_{\delta} \right) \\
&\leq \sum_{j=2}^{t} \sum_{\mathbf{Y}\in A} \frac{h(\delta)}{q^{\frac{k}{\log k}-\delta}} \label{e:c2} \\
&< \frac{t\cdot h(\delta)}{q^{c-\delta}}. \label{e:b}
\end{align}
Furthermore,
\begin{align}
Pr(F) &= \sum_{p_1,\ldots,p_{\delta}\in GF(q)} Pr\left(F| p_1,\ldots,p_{\delta}\right)Pr\left( p_1,\ldots,p_{\delta} \right) \label{e:111} \\
&< \sum_{p_1,\ldots,p_{\delta}\in GF(q)} \frac{t \times h(\delta)}{q^{c-\delta}}~Pr\left( p_1,\ldots,p_{\delta} \right) \label{e:122}\\
&=  \frac{t \cdot h(\delta)}{q^{c-\delta}}. \label{e:a}
\end{align}
Therefore,
\begin{equation}
Pr(F) = \mathcal{O} \left( \frac{1}{k^{c-2\delta}~\log^{\delta} k} \right). \label{e:t}
\end{equation}
\eqref{e:c2} follows from Claim~\ref{claim:2}. \eqref{e:122} follows from \eqref{e:b}. \eqref{e:t} follows from \eqref{ttt}, \eqref{e:a} and the fact that $h(\delta)$ is constant (independent of $k$) for a constant $\delta$.

\subsection{Proof of Claim~\ref{claim:1}} 
\label{clm1}
We focus on case $i$ ($i$ fixed) that assumes that the deletion has occurred in block $i$. We observe the output $\mathcal{Y}_i$ of the decoder in step~$3$ of the decoding scheme for all possible input messages, for a fixed deletion position and a given parity $p_1$. Recall that $\bm{\mathcal{Y}_i}$ is a random variable taking values in $\mathrm{A_1}$. Claim~\ref{claim:1} gives an upper bound on the probability mass function of $\bm{\mathcal{Y}_i}$ for any $i$ and for a given $p_1$. We distinguish here between two cases. If case $i$ is correct, i.e., the  assumption on the deletion position is correct, then $\bm{\mathcal{Y}_i}$ is always decoded correctly and it is uniformly distributed over $\mathrm{A_1}$. If case $i$ is wrong, then $\bm{\mathcal{Y}_i}$ is not uniformly distributed over $\mathrm{A_1}$ as illustrated in the next example.
\begin{example} 
\label{ex:3}
Let the length of the binary message $\mathbf{u}$ be $k=16$ and consider $\delta=1$ deletion. Let the first MDS parity $p_1$ be the sum of the $k/\log k=4$ message symbols. Consider the previously defined set $\mathrm{A_1}$ with $p_1=0$. Consider the messages,
\begin{align*}
\mathbf{U_1} &= (0,0,0,0) \in \mathrm{A_1}, \\
\mathbf{U_2} &= (\alpha,0,0,\alpha) \in \mathrm{A_1}.
\end{align*}
For the sake of simplicity, we skip the last encoding step (Block IV, Fig.~\ref{fig:1}), and assume that $p_1$ is recovered at the decoder. Therefore, the corresponding codewords to be transmitted are
\begin{align*}
\mathbf{x_1} &= \color{blue}0~0~\underline{\textbf{0}}~0~0~0~0~0~0~0~0~0~0~0~0~0~\color{red}0~0~0~0\color{black}, \\
\mathbf{x_2} &= \color{blue}0~0~\underline{\textbf{1}}~0~0~0~0~0~0~0~0~0~0~0~1~0~\color{red}0~0~0~0\color{black}.
\end{align*}
Now, assume that the $3^{rd}$ bit of $\mathbf{x_1}$ and $\mathbf{x_2}$ was deleted, and case $4$ (wrong case) is considered. It easy to verify that in this case, for both codewords, the $q-$ary output of the decoder will be
\begin{equation*}
\bm{\mathcal{Y}_4} = (0,0,0,0) \in \mathrm{A_1}.
\end{equation*}
This shows that there exists a wrong case $i$, where the same output can be obtained for two different inputs and a fixed deletion position. Thereby, the distribution of $\bm{\mathcal{Y}_i}$ over $\mathrm{A_1}$ is not uniform.
\end{example}
The previous example suggests that to find the bound in Claim~\ref{claim:1}, we need to determine the maximum number of different inputs that can generate the same output for an arbitrary fixed deletion position and a given parity $p_1$. We call this number $\gamma$. Once we obtain $\gamma$ we can write
\begin{equation}
Pr\left(\bm{\mathcal{Y}_i}=\mathbf{Y} | \mathcal{D}=d, \mathcal{P}_1=p_1 \right) \leq \frac{\gamma}{\left \lvert \mathrm{A_1} \right \rvert} = \frac{\gamma}{q^{\frac{k}{\log k}-1}}, \label{gamma}
\end{equation}
where $\mathcal{D}\in \{1,\ldots,n\}$ is the random variable representing the position of the deleted bit. We will explain our approach for determining $\gamma$ by going through an example for $k=16$ that can be easily generalized for any $k$. We denote by $b_z\in GF(2)$, $z=1,2,\ldots,k$, the bit of the message $\mathbf{u}$ in position $z$.
\begin{example} \label{ex:4}
Let $k=16$ and $\delta=1$. Consider the binary message $\mathbf{u}$ given by
\begin{equation*}
\mathbf{u}=b_1~b_2~b_3~b_4~b_5~b_6~b_7~b_8~b_9~b_{10}~b_{11}~b_{12}~b_{13}~b_{14}~b_{15}~b_{16}.
\end{equation*}
The extension field used here has a primitive element $\alpha$, with $\alpha^4=\alpha+1$. Assume that $b_3$ was deleted and case $4$ is considered. Hence, the binary string received at the decoder is chunked as follows
\begin{equation*}
\underbrace{b_1~b_2~b_4~b_5}_{S_1}~\underbrace{b_6~b_7~b_8~b_9}_{S_2}~\underbrace{b_{10}~b_{11}~b_{12}~b_{13}}_{S_3}~\underbrace{b_{14}~b_{15}~b_{16}}_{\mathcal{E}},
\end{equation*}
where the erasure is denoted by $\mathcal{E}$, and $S_1, S_2$ and $S_3$ are the first 3 symbols of $\bm{\mathcal{Y}_4}$ given by
\begin{align*}
S_1 &= \alpha^3 b_1+\alpha^2 b_2 +\alpha b_4+ b_5 \in GF(16), \\
S_2 &= \alpha^3 b_6+\alpha^2 b_7 +\alpha b_8+ b_9 \in GF(16),\\
S_3 &= \alpha^3 b_{10}+\alpha^2 b_{11} +\alpha b_{12}+ b_{13} \in GF(16).
\end{align*}
The fourth symbol $S_4$ of $\bm{\mathcal{Y}_4}$ is to be determined by erasure decoding.
Suppose that there exists another message $\mathbf{u'}\neq \mathbf{u}$ such that $\mathbf{u}$ and $\mathbf{u'}$ lead to the same decoded string \mbox{$\bm{\mathcal{Y}_4}=(S_1,S_2,S_3,S_4)$}. Since these two messages  generate the same values of $S_1$, $S_2$ and $S_3$, then they should have the same values for the following bits 
\begin{equation*}
b_1~b_2~b_4~b_5~b_6~b_7~b_8~b_9~b_{10}~b_{11}~b_{12}~b_{13}.
\end{equation*}
We refer to these $k-\log k=12$ bits by ``fixed" bits. The only bits that can be different in $\mathbf{u}$ and $\mathbf{u'}$ are $b_{14}$, $b_{15}$ and $b_{16}$ which correspond to the erasure, and the deleted bit $b_3$. We refer to these $\log k=4$ bits by ``free" bits. Although these ``free" bits can be different in $\mathbf{u}$ and $\mathbf{u'}$, they are constrained by the fact that the first parity in their corresponding $q-$ary codewords $\mathbf{X}$ and $\mathbf{X'}$ should have the same value $p_1$. Next, we express this constraint by a linear equation in $GF(16)$. Without loss of generality, we assume that $p_1 \in GF(16)$ is the sum of the $k/\log k=4$ message symbols. Hence, $p_1$ is given by
\begin{align*}
p_1=\alpha^3 & \left(b_1+b_5+b_9+b_{13}\right)+\alpha^2\left(b_2+b_6+b_{10}+b_{14}\right)  \\
&+\alpha\left(b_3+b_7+b_{11}+b_{15} \right)+\left(b_4+b_8+b_{12}+b_{16}\right).
\end{align*}
Rewriting the previous equation by having the ``free" bits on the LHS and the ``fixed" bits and $p_1$ on the RHS we get
\begin{equation}
\label{eq:22}
\alpha b_3 + \alpha^2 b_{14} + \alpha b_{15} + b_{16} = p',
\end{equation} 
where $p'\in GF(16)$ and is given by $p'=p_1+\alpha^3\left(b_1+b_5+b_9+b_{13}\right)+\alpha^2\left(b_2+b_6+b_{10}\right)+\alpha\left(b_7+b_{11} \right)+\left(b_4+b_8+b_{12}\right)$. The previous equation can be written as the following linear equation in $GF(16)$,
\begin{equation}
\label{eq:11}
0\alpha^3+b_{14}\alpha^2+ \left(b_3+b_{15}\right)\alpha  + b_{16} = p'.
\end{equation}
Now, to determine $\gamma$, we count the number of solutions of \eqref{eq:11}. If the unknowns in \eqref{eq:11} were symbols in $GF(16)$, then the solutions of \eqref{eq:11} would span an affine subspace of size $q^{k/\log k-1}=16^3$. However, the unknowns in \eqref{eq:11} are binary, so we show next that it has at most $2$ solutions. Let 
\begin{equation}
a_3\alpha^3+a_2\alpha^2+a_1\alpha+a_0=p'\label{e:p}
\end{equation}
be the polynomial representation of $p'$ in $GF(16)$ where $(a_3,a_2,a_1,a_0)\in GF(2)^4$. Every element in $GF(16)$ has a unique polynomial representation of degree at most $3$. Comparing \eqref{eq:11} and \eqref{e:p}, we obtain the following system of equations
\[
  \left\{\def\arraystretch{1.2}%
  \begin{array}{@{}c@{\quad}l@{}}
    b_{16}&=a_0,\\
    b_3+b_{15}&=a_1,\\
    b_{14}&=a_2,\\
    0&=a_3.
  \end{array}\right. 
\]
If $a_3\neq 0$, then \eqref{eq:11} has no solution. If $a_3= 0$, then \eqref{eq:11} has $2$ solutions because $b_3+b_{15}=a_1$ has $2$ solutions. Therefore, \eqref{eq:11} has at most $2$ solutions, i.e., $\gamma \leq 2$.
\end{example}
The analysis in Example~\ref{ex:4} can be directly generalized for messages of any length $k$. In general, the analysis yields $\log k$ ``free" bits and $k-\log k$ ``fixed" bits. Now, we generalize \eqref{eq:11} and show that $\gamma\leq 2$ for any $k$. Without loss of generality, we assume that $p_1\in GF(q)$ is the sum of the $k/\log k$ symbols of the $q-$ary message $\mathbf{U}$. Consider a wrong case $i$ that assumes that the deletion has occurred in block $i$. Let $d_j$ be a fixed bit position in block $j$, $j\neq i$, that represents the position of the deletion. Depending on whether the deletion occurred before or after block $i$, the generalization of \eqref{eq:11} is given by one of the two following equations in $GF(q)$. \\
If $j<i$,
\begin{equation}
b_{d_j} \alpha^w  + b_{\ell+1}\alpha^{m-1} + b_{\ell}\alpha^{m-2}  + \ldots + b_{\ell+m} = p'', \label{e:ji}
\end{equation}
If $j>i$,
\begin{equation}
b_{d_j}\alpha^w  + b_{\ell}\alpha^m + b_{\ell+1}\alpha^{m-1}  + \ldots + b_{\ell+m-1}\alpha  = p'', \label{e:ij}
\end{equation}
\noindent where $\ell=(i-1)\log k+1$, $m=\log k-1$, $w=j\log k-b_j$ and $p''\in GF(q)$ (the generalization of $p'$ in Example~\ref{ex:4}) is the sum of $p_1$ and the part corresponding to the ``fixed" bits. Suppose that $j<i$. Note that $1\leq w \leq m$, so \eqref{e:ji} can be written as
\begin{equation}
b_{\ell+1} \alpha^{m-1} + \ldots + (b_{d_j}+b_{\ell+o})\alpha^{w}  + \ldots + b_{\ell+m} = p'', \label{e:ji2}
\end{equation}
where $o$ is an integer such that $1 \leq o \leq m$. Hence, by the same reasoning used in Example~\ref{ex:4} we can conclude that \eqref{e:ji2} has at most $2$ solutions. The same reasoning applies for \eqref{e:ij}, where $j>i$. Therefore, $\gamma \leq 2$ and from \eqref{gamma} we have
\begin{equation}
Pr\left(\bm{\mathcal{Y}_i}=\mathbf{Y} | \mathcal{D}=d, \mathcal{P}_1=p_1\right) \leq \frac{2}{q^{\frac{k}{\log k}-1}}. \label{c:1}
\end{equation}
The bound in~\eqref{c:1} holds for arbitrary $d$. Therefore, the upper bound on the probability of decoding failure in~\eqref{d:1} holds for any deletion position picked independently of the codeword. Moreover, for any given distribution on $\mathcal{D}$ (like the uniform distribution for example), we can apply the total law of probability with respect to $\mathcal{D}$ and use the result from~\eqref{c:1} to get
\begin{equation*}
Pr\left(\bm{\mathcal{Y}_i}=\mathbf{Y} | \mathcal{P}_1=p_1 \right) \leq  \frac{2}{q^{\frac{k}{\log k}-1}}.
\end{equation*}
\section{Trade-offs}
\label{sec:6.1}
As previously mentioned, the first encoding step in GC codes (Block I, Fig.~\ref{fig:1}) consists of chunking the message into blocks of length $\log k$ bits. In this section, we generalize the results in Theorem~\ref{thm:1} by considering chunks of arbitrary length $\ell$ bits $(\ell \leq k)$ \footnote{For $\ell=k$ and $c=1$, the code is equivalent to a $(\delta +1)$ repetition code.}, instead of $\log k$ bits. We show that if $\ell=\Omega(\log k)$, then GC codes have an asymptotically vanishing probability of decoding failure. This generalization allows us to demonstrate two trade-offs achieved by GC codes, based on the code properties in Theorem~\ref{thm:2}.
\begin{theorem} 
\label{thm:2} 
Guess \& Check (GC) codes can correct in polynomial time up to a constant number of $\delta$ deletions. Let $c>\delta$ be a constant integer. The code has the following properties:
\begin{enumerate}[leftmargin=*]
\item Redundancy: $n-k = c(\delta+1) \ell$ bits. 
\item Encoding complexity is $\mathcal{O}(k \ell)$, and  decoding complexity is $\mathcal{O}\left(\frac{k^{\delta+1}}{\ell^{\delta-1}}\right)$.
\item For any $\delta$ deletion positions chosen independently of the codeword, the probability that the decoding fails for a uniform iid message is: $Pr(F) = \mathcal{O}\left(\frac{(k/\ell)^{\delta}}{2^{\ell(c-\delta)}} \right)$. 
\end{enumerate}
\end{theorem}
\begin{proof}
See Appendix~\ref{app:B}.
\end{proof}
\noindent  The code properties in Theorem~\ref{thm:2} enable two trade-offs for GC codes:

\noindent $1)$ {\em Trade-off A:}
By increasing $\ell$ for a fixed $k$ and $c$, we observe from Theorem~\ref{thm:2} that the redundancy increases linearly while the decoding complexity decreases as a polynomial of degree~$\delta$. Moreover, increasing $\ell$ also decreases the probability of decoding failure as a polynomial of degree~$\delta$. Therefore, trade-off A shows that by paying a linear price in terms of redundancy, we can simultaneously gain a~degree $\delta$ polynomial improvement in both decoding complexity and probability of decoding failure.

\noindent $2)$ {\em Trade-off B:}
By increasing $c$ for a fixed $k$ and $\ell$, we observe from Theorem~\ref{thm:2} that the redundancy increases linearly while the probability of decoding failure decreases exponentially. Here, the order of complexity is not affected by $c$. Therefore, trade-off B shows that by paying a linear price in terms of redundancy, we can gain an exponential improvement in probability of decoding failure, without affecting the order of complexity. This trade-off is of particular importance in models which allow feedback, where asking for additional parities (increasing $c$) will highly increase the probability of successful decoding. 
\section{Correcting $\delta$ Insertions}
\label{sec:i}
In this section we show that by modifying the decoding scheme of GC codes, and keeping the same encoding scheme, we can obtain codes that can correct $\delta$ insertions, instead of $\delta$ deletions. In fact, the resulting code properties for $\delta$ insertions (Theorem~\ref{thm:3}), are the same as that of $\delta$ deletions. Recall that $\ell$ (Section~\ref{sec:6.1}) is the code parameter representing the chunking length, i.e., the number of bits in a single block. For correcting $\delta$ insertions, we keep the same encoding scheme as in Fig.~\ref{fig:1}, and for decoding we only modify the following:~(i)~Consider the assumption that a certain block $B$ is affected by $\delta'$ insertions, then while decoding we chunk a length of $\ell+\delta'$ at the position of block $B$ (compared to chunking $\ell-\delta'$ bits when decoding deletions); (ii)~The blocks assumed to be affected by insertions are considered to be erased (same as the deletions problem), but now for Criterion 2 (Definition~\ref{def:4}) we check if the decoded erasure is a subsequence of the chunked super-block (of length $\ell+\delta'$). The rest of the details in the decoding scheme stay the same. 
\begin{theorem} 
\label{thm:3} 
Guess \& Check (GC) codes can correct in polynomial time up to a constant number of $\delta$ insertions. Let $c>\delta$ be a constant integer. The code has the following properties:
\begin{enumerate}[leftmargin=*]
\item Redundancy: $n-k = c(\delta+1) \ell$ bits.
\item Encoding complexity is $\mathcal{O}(k \ell)$, and  decoding complexity is $\mathcal{O}\left(\frac{k^{\delta+1}}{\ell^{\delta-1}}\right)$.
\item For any $\delta$ insertion positions chosen independently of the codeword, the probability that the decoding fails for a uniform iid message is: $Pr(F) = \mathcal{O}\left(\frac{(k/\ell)^{\delta}}{2^{\ell(c-\delta)}} \right)$. 
\end{enumerate}
\end{theorem}
We omit the proof of Theorem~\ref{thm:3} because the same analysis applies as in the proof of Theorem~\ref{thm:1}. More specifically, the redundancy and the encoding complexity are the same as in Theorem~\ref{thm:1} because we use the same encoding scheme. The decoding complexity is also the same as in Theorem~\ref{thm:1} because the total number of cases to be checked by the decoder is unchanged. The proof of the upper bound on the probability of decoding failure also applies similarly using the same techniques.
\section{Simulation Results}
\label{sec:7}

We simulated the decoding  of GC codes and compared the obtained probability of decoding failure  to the  upper bound in Theorem~\ref{thm:1}. We tested the code for messages of length $k~=~256, 512$ and $1024$ bits, and for $\delta=2, 3$ and $4$ deletions.
\begin{table}[h]
\centering
\setlength\extrarowheight{1.2pt}
 \begin{tabular}{|c|c|c|c|c|c|c|c|c|}
\hline
\multirow{2}{*}{Config.} & \multicolumn{6}{c|}{$\delta$} \\ \cline{2-7}
& \multicolumn{2}{c|}{$2$}  & \multicolumn{2}{c|}{$3$} &  \multicolumn{2}{c|}{$4$} \\ \hline
$k$ & $R$ & $Pr(F)$ & $R$ & $Pr(F)$ &  $R$ & $Pr(F)$ \\ \hline
256 & $0.78$ & $1.3e^{-3}$ & $0.67$ & $4.0e^{-4}$ & $0.56$ &$0$ \\ \hline
512 & $0.86$ & $3.0e^{-4}$ & $0.78$ &$0$ & $0.69$ &$0$  \\ \hline
1024 & $0.92$ & $2.0e^{-4}$ & $0.86$ &$0$ & $0.80$ &$0$  \\ \hline
\end{tabular}
\captionsetup{font=footnotesize}
\caption{\footnotesize{The table shows the code rate $R=k/n$ and the probability of decoding failure $Pr(F)$  of GC codes for different message lengths $k$ and different number of deletions $\delta$. The results shown are for $c=\delta+1$ and $\ell=\log k$. The results of $Pr(F)$ are averaged over $10000$ runs of simulations. In each run, a message $\mathbf{u}$ chosen uniformly at random is encoded into the codeword $\mathbf{x}$. $\delta$ bits are then deleted uniformly at random from $\mathbf{x}$, and the resulting string is decoded.
}}
\vspace{-0.2cm}
\label{t}
\end{table}
To guarantee an asymptotically vanishing probability of decoding failure, the upper bound in Theorem~\ref{thm:1} requires that $c>2\delta$. Therefore, we make a distinction between two regimes, \mbox{(i) $\delta<c<2\delta:$}~Here, the theoretical upper bound is trivial. Table~\ref{t} gives the results for $c=\delta+1$ with the highest probability of decoding failure observed in our simulations  being of the order of $10^{-3}$. This indicates that GC codes can decode correctly with high probability in this regime, although not reflected in the upper bound;  \mbox{(ii) $c>2\delta:$}  The upper bound is of the order of $10^{-5}$ for $k=1024,\delta=2$, and $c=2\delta +1$. In the simulations no decoding failure was detected within $10000$ runs for $\delta+2 \leq c \leq 2\delta+1$. 
In general, the simulations show that GC codes  perform better than what the upper bound indicates. This is due to the fact that the effect of Criterion 2 (Definition~\ref{def:4}) is not taken into account when deriving the upper bound in Theorem~\ref{thm:1}. These simulations were performed on a personal computer and the programming code was not optimized. The average decoding time\footnote{These results are for $\ell=\log k$, the decoding time can be decreased if we increase $\ell$ (trade-off~A,  Section~\ref{sec:6.1}).} is in the order of milliseconds  for $(k=1024,\delta=2)$, order of seconds for $(k=1024,\delta=3)$, and order of minutes for $(k=1024,\delta=4)$. Going beyond these values of $k$ and $\delta$ will largely increase the running time due to the number of cases to be tested by the decoder. However, for the file synchronization application in which we are interested (see next section) the values of $k$ and $\delta$ are relatively small and decoding can be practical. 

\section{Application to File Synchronization}
\label{sec:8}
In this section, we describe how GC codes can be used to construct interactive protocols for file synchronization. We consider the model where two nodes (servers) have copies of the same file but one is obtained from the other by deleting $d$ bits. These nodes communicate interactively over a noiseless link to synchronize the file affected by deletions. Some of the most recent work on synchronization can be found in\mbox{\cite{V15,R10,Y14,S16,M11}}. In this section, we modify the synchronization algorithm by Venkataramanan {\em et al.} \cite{V15,R10}, and study the improvement that can be achieved by including GC codes as a black box inside the algorithm. The key idea in \cite{V15,R10} is to use {\em center bits} to divide a large string, affected by $d$ deletions, into shorter segments, such that each segment is affected by one deletion at most. Then, use VT codes to correct these shorter segments. Now, consider a similar algorithm where the large string is divided such that the shorter segments are affected by \mbox{$\delta$ $(1<\delta\ll d)$} or fewer deletions. Then, use GC codes to correct the segments affected by more than one deletion\footnote{VT codes are still used for segments affected by only one deletion.}. We set $c=\delta+1$ and $\ell=\log k$, and if the decoding fails for a certain segment, we send one extra MDS parity at a time within the next communication round until the decoding is successful. By implementing this algorithm, the gain we get is two folds: (i)~reduction in the number of communication rounds; (ii)~reduction in the total communication cost. We performed simulations for $\delta=2$ on files of size 1 Mb, for different numbers of deletions $d$. The results are illustrated in Table~\ref{ts}. We refer to the original scheme in \cite{V15,R10} by {\em Sync-VT}, and to the modified version by {\em Sync-GC}. The savings for $\delta=2$ are roughly $43\%$ to $73\%$ in number of rounds, and $5\%$ to $14\%$ in total communication cost. 
\begin{table}[h]
\centering
\setlength\extrarowheight{1.2pt}
 \begin{tabular}{|c|c|c|c|c|}
\hline
 & \multicolumn{2}{c|}{Number of rounds} & \multicolumn{2}{c|}{Total communication cost}\\ \hline
$d$ & Sync-VT & Sync-GC & Sync-VT & Sync-GC  \\ \hline
100 & $14.52$ & $10.15$ & $5145.29$ & $4900.88$ \\ \hline
150 & $16.45$ & $10.48$ & $7735.32$ & $7199.20$  \\ \hline
200 & $17.97$ & $10.88$ & $10240.60$ & $9332.68$  \\ \hline
250 & $18.93$ & $11.33$ & $12785.20$ & $11415.90$ \\ \hline
300 & $20.29$ & $11.70$ & $15318.20$ & $13397.80$ \\ \hline
\end{tabular}
\captionsetup{font=footnotesize}
\caption{ Results are averaged over $1000$ runs. In each run, a string of size 1~Mb is chosen uniformly at random, and the file to be synchronized is obtained by deleting $d$ bits from it uniformly at random. The total communication cost is expressed in bits. The number of {\em center bits} used is 25.}
\vspace{-0.3cm}
\label{ts}
\end{table}

Note that the interactive algorithm in~\cite{V15,R10} also deals with the general file synchronization problem where the edited file is affected by both deletions and insertions. There, the string is divided into shorter segments such that each segment is either:~(i)~not affected by any deletions/insertions; or~(ii) affected by only one deletion; or (iii)~affected by only one insertion. Then, VT codes are used to correct the short segments. GC codes can also be used in a similar manner when both deletions and insertions are involved. In this case, the string would be divided such that the shorter segments are affected by:~(i)~$\delta$ or fewer deletions; or (ii)~$\delta$ or fewer insertions.
\section{Conclusion}
\label{sec:11}
In this paper, we introduced a new family of systematic codes, that we called Guess \& Check (GC) codes, that can correct multiple deletions (or insertions) with high probability. We provided deterministic polynomial time encoding and decoding schemes for these codes. We validated our theoretical results by numerical simulations. Moreover, we showed how these codes can be used in applications to remote file synchronization.
In conclusion, we point out some open problems and possible directions of future work:
\begin{enumerate}
\item GC codes can correct $\delta$ deletions or $\delta$ insertions. Generalizing these constructions to deal with mixed deletions and insertions ({\em indels}) is a possible future direction.
\item In our analysis, we declare a decoding failure if the decoding yields more than one possible candidate string. A direction is to study the performance of GC codes in list decoding.
\item Developing faster decoding algorithms for GC codes is also part of the future work.
\item From the proof of Theorem~\ref{thm:1}, it can be seen that bound on the probability of decoding failure of GC codes holds if the deletion positions are chosen by an adversary that does not observe the codeword. It would be interesting to study the performance of GC codes under more powerful adversaries which can observe part of the codeword, while still allowing a vanishing probability of decoding failure. 
\end{enumerate}

\appendices
\section{Proof of Claim \ref{claim:2}}
\label{app:A}
Claim~\ref{claim:2} is a generalization of Claim~\ref{claim:1} for $\delta>1$ deletions. Recall that the decoder goes through $t$ cases where in each case corresponds to one possible way to distribute the $\delta$ deletions among the $k/\log k$ blocks. Claim~\ref{claim:2} states that there exists a deterministic function $h$ of $\delta$, $h(\delta)$ independent of $k$, such that for any case $i$, $i=1,2,\ldots,t$,
\begin{equation*}
Pr\left(\bm{\mathcal{Y}_i}=\mathbf{Y} | \mathcal{P}_1=p_1,\ldots,\mathcal{P}_{\delta}=p_{\delta}\right) \leq \frac{h(\delta)}{q^{\frac{k}{\log k}-\delta}},
\end{equation*} 
where $\bm{\mathcal{Y}_i}$ is the random variable representing the $q-$ary string decoded in case $i$, $i=1,2,\ldots,t$, and $\mathcal{P}_r, r=1,\ldots,c$, is the random variable representing the $r^{th}$ MDS parity symbol. 

To prove the claim, we follow the same approach used in the proof of Claim~\ref{claim:1}. Namely, we count the maximum number of different inputs (messages) that can generate the same output (decoded string) for $\delta$ fixed deletion positions ($d_1,\ldots,d_{\delta}$) and $\delta$ given parities ($p_1,\ldots,p_{\delta}$). Again, we call this number $\gamma$. Recall that, for $\delta$ deletions, $\bm{\mathcal{Y}_i}$ is decoded based on the first $\delta$ parities\footnote{The underlying assumption here is that the $\delta$ deletions affected exactly $\delta$ blocks. In cases where it is assumed that less than $\delta$ blocks are affected, then less than $\delta$ parities will be used to decode $\bm{\mathcal{Y}_i}$, and the same analysis applies.}. Hence, $\bm{\mathcal{Y}_i} \in \mathrm{A^{\delta}}$, where
\begin{align}
\mathrm{A^{\delta}}&\triangleq \mathrm{A_1} \cap \mathrm{A_2} \cap \ldots \cap \mathrm{A_\delta}, \\
\mathrm{A_r} &\triangleq \{ \mathbf{Y} \in GF(q)^{k/\log k} |~\mathbf{G_r^TY}=p_r\} \label{eq:A1},
\end{align}
for $r=1,\ldots,\delta$ \footnote{The set $\mathrm{A^{\delta}}$ is the generalization of set $\mathrm{A_1}$ for $\delta=1$.}. We are interested in showing that $\gamma$ is independent of the binary message length $k$. To this end, we upper bound $\gamma$ by a deterministic function of $\delta$ denoted by $h(\delta)$. Hence, we establish the following bound
\begin{equation}
\label{back}
Pr\left(\bm{\mathcal{Y}_i}=\mathbf{Y} | d_1,\ldots, d_{\delta},p_1,\ldots,p_{\delta}\right) \leq \frac{\gamma}{\left \lvert \mathrm{A^{\delta}} \right \rvert} \leq \frac{h(\delta)}{q^{\frac{k}{\log k}-\delta}}.
\end{equation}
We will now explain the approach for bounding $\gamma$ through an example for $\delta=2$ deletions.
\begin{example}
\label{ex:5}
Let $k=32$ and $\delta=2$. Consider the binary message $\mathbf{u}$ given by
\begin{equation*}
\mathbf{u}=b_1~b_2~\ldots~b_{32}.
\end{equation*}
Its corresponding $q-$ary message $\mathbf{U}$ consists of $7$ symbols (blocks) of length $\log k=5$ bits each. The message $\mathbf{u}$ is encoded into a codeword $\mathbf{x}$ using the GC code (Fig.~\ref{fig:1}). We assume that the first parity is the sum of the systematic symbols and the encoding vector for the second parity is $(1,\alpha,\alpha^2,\alpha^3,\alpha^4,\alpha^5,\alpha^6)$ \footnote{The extension field used is $GF(32)$ and has a primitive element $\alpha$, with $\alpha^5=\alpha^2+1$.}. Moreover, we assume that the actual deleted bits in $\mathbf{x}$ are $b_1$ and $b_7$, but the chunking is done based on the assumption that deletions occurred in the $3^{rd}$ and $5^{th}$ block (wrong case). Similar to Example~\ref{ex:4}, it can be shown that the ``free" bits are constrained by the following system of two linear equations in $GF(32)$,
\begin{equation}
\label{ed1}
\left\{
\begin{array}{l}
\alpha^4 b_1 + \alpha^3 b_7 +\alpha^2 b_{13}+\alpha^4 b_{14}+b_{15}+\alpha^4 b_{16}  \\
~~~~~~~~~~~~~~~~~~~~~~+\alpha^3 b_{22} + \alpha^2 b_{23} + \alpha b_{24} +b_{25} = p'_1,  \\
~~\\
\alpha^4 b_1 +\alpha^4 b_7 +\alpha^2(\alpha^2 b_{13}+\alpha^4 b_{14}+b_{15})+\alpha^7 b_{16}  \\
~~~~~~~~~~~~~~~~~+\alpha^5(\alpha^3 b_{22} + \alpha^2 b_{23} + \alpha b_{24} +b_{25}) = p'_2. 
\end{array}
\right.
\end{equation}
To upper bound $\gamma$, we upper bound the number of solutions of the system given by \eqref{ed1}. Equation~\eqref{ed1} can be written as follows
\begin{equation}
\label{ed3}
  \left\{
  \begin{array}{rrr}
    (b_1+b_{16}) \alpha^4 + b_7 \alpha^3 + B_1 + B_2 = p'_1, \\
   (b_1+b_7) \alpha^4 +  b_{16}\alpha^7 + \alpha^2 B_1 + \alpha^5 B_2 = p'_2, 
  \end{array}
  \right.
\end{equation}
where $B_1$ and $B_2$ are two symbols in $GF(32)$ given by
\begin{align}
 B_1 &= \alpha^2 b_{13}+\alpha^4 b_{14}+b_{15}, \label{p1}\\
B_2 &= \alpha^3 b_{22} + \alpha^2 b_{23} + \alpha b_{24} +b_{25}. \label{p2}
\end{align} 
Notice that the coefficients of $B_1$ and $B_2$ in~\eqref{ed3} originate from the MDS encoding vectors. Hence, for given bit values of $b_1$, $b_7$ and $b_{16}$, the MDS property implies that \eqref{ed3} has a unique solution for $B_1$ and $B_2$. Furthermore, since $B_1$ and $B_2$ have unique polynomial representations in $GF(32)$ of degree at most 4, for given values of $B_1$ and $B_2$, \eqref{p1} and \eqref{p2} have at most one solution for $b_{13}, b_{14}, b_{15}, b_{22}, b_{23}, b_{24}$ and $b_{25}$. We think of bits $b_{13}, b_{14}, b_{15}, b_{22}, b_{23}, b_{24}$ and $b_{25}$ as ``free" bits of type~I, and bits $b_{1},b_{7},b_{16}$ as ``free" bits of type~II. The previous reasoning indicates that for given values of the bits of type~II, \eqref{ed1} has at most one solution. Therefore, an upper bound on $\gamma$ is given by the number of possible choices of the bits of type~II. Hence, $\gamma \leq 2^3=8$.
\end{example}
\noindent Now, we generalize the previous example and upper bound $\gamma$ for $\delta>2$ deletions. Without loss of generality, we assume that the $\delta$ deletions occur in $\delta$ different blocks. Then, the ``free" bits are constrained by  a system of $\delta$ linear equations in $GF(q)$, where $q=k$. Let $\nu_{(.)}$ and $\mu_{(.)}$ be non-negative integers of value at most $k$. Each of the $\delta$ equations has the following form 
\begin{equation}
\sum_{i=1}^{\beta} \alpha^{\nu_i} b_{\mu_i} + \sum_{j=1}^{\delta} \alpha^{\nu_j} B_j=p',
\end{equation}
where $\beta$ is the number of ``free" bits of type~II, $B_j\in GF(q)$ is a linear combination of part of the bits in block $j$ (``free" bits of type~I), and $p'\in GF(q)$. The coefficients of $B_j$, $j=1,\ldots, \delta$, originate from the MDS code generator matrix. Hence, for given values of the bits of type~II, the system of $\delta$ equations has a unique solution for $B_j$, $j=1,\ldots,\delta$. Furthermore, the linear combination of the bits that gives $B_j$, $j=1,\ldots,\delta$, has the following form
\begin{equation}
\label{Bj}
B_j=\alpha^m b_{j_1} + \alpha^{m-1} b_{j_2} + \ldots + \alpha^{m-\lambda_j+1} b_{j_{\lambda_j}},
\end{equation} 
where $m<\log k$ is an integer, and $\lambda_j$ is the number of ``free" bits of type~I in $B_j$. Notice that \eqref{Bj} corresponds to a polynomial representation in $GF(q)$, $q=k$, of degree less than $\log k$. Hence, for a given $B_j$, \eqref{Bj} has at most one solution. Therefore, $\gamma$ is upper bounded by the number of possible choices of the bits of type~II, i.e., $\gamma \leq 2^{\beta}$. Recall that, when it is assumed that the block $j$ is affected by deletions, a sub-block of bits is chunked at the position of block $j$. However, because of the shift caused by the $\delta$ deletions, that sub-block may contain bits which do not originate from block~$j$. The $\lambda_j$ ``free" bits of type~I in $B_j$, $b_{j_1},\ldots,b_{j_{\lambda_j}}$, are the bits of the sub-block which do originate from block $j$. Since the shift at block $j$ is at most $\delta$ positions, it is easy to see that for any $j=1,\ldots,\delta$, we have $\lambda_j \geq \log k - \delta$. Therefore,
since $\beta= \delta \log k - \sum_{j=1}^{\delta} \lambda_j$, we can use the lower bound on $\lambda_j$ to show that $\beta\leq \delta^2$. Hence, we obtain $\gamma \leq 2^{\delta^2} \triangleq h(\delta)$~\footnote{If the $\delta$ deletions occur in $z<\delta$ blocks, then $\beta=z\log k - \sum_{j=1}^{z}\lambda_j$, and the upper bound $\beta \leq \delta^2$ would still hold.}. In summary, we have shown that $\gamma$ is upper bounded by a deterministic function of $\delta$ that is independent of $k$. 

Since the bound in~\eqref{back} holds for arbitrary deletion positions ($d_1,\ldots,d_{\delta}$), the upper bound on the probability of decoding failure in Theorem~\ref{thm:1} holds for any $\delta$ deletion positions picked {\em independently} of the codeword. Moreover, for any given random distribution on the $\delta$ deletion positions (like the uniform distribution for example), we can apply the total law of probability and use the result from~\eqref{back} to get
\begin{equation*}
Pr\left(\bm{\mathcal{Y}_i}=\mathbf{Y} | \mathcal{P}_1=p_1,\ldots, \mathcal{P}_{\delta}=p_{\delta} \right) \leq \frac{h(\delta)}{q^{\frac{k}{\log k}-\delta}}.
\end{equation*}
\section{Proof of Theorem~\ref{thm:2}}
\label{app:B}
We consider the same encoding scheme illustrated in Fig.~\ref{fig:1} with the only modification that the message is chunked into blocks of length $\ell$ bits, and the field is size is $q=2^{\ell}$. Taking this modification into account, the proof of Theorem~\ref{thm:2} follows the same steps of the proof of Theorem~\ref{thm:1}. It is easy to see from Fig.~\ref{fig:1} that the redundancy here becomes $c(\delta+1)\ell$ bits. Also, the number of $q-$ary systematic symbols becomes $k/\ell$. Therefore, the total number of cases to be checked by the decoder is
\begin{equation}
t=\binom{k/\ell+\delta-1}{\delta}=\mathcal{O}\left(\frac{k^{\delta}}{\ell^{\delta}}\right).
\end{equation}
Furthermore, from the proof of Theorem~\ref{thm:1} we have that the encoding complexity is
\begin{equation}
\mathcal{O}\left(c \cdot \frac{k}{\ell}\cdot \log^2 q \right)=\mathcal{O}\left(k\ell\right),
\end{equation}
and the decoding complexity is
\begin{equation}
\mathcal{O}\left(t \cdot \log^2 q \cdot k/\ell \right)=\mathcal{O}\left(\frac{k^{\delta+1}}{\ell^{\delta-1}} \right).
\end{equation}
As for the probability of decoding failure, the same intermediary steps of the proof of Theorem~\ref{thm:1} apply for chunks of length $\ell$ instead of $\log k$. In particular, the key property used in the proofs of Claim~\ref{claim:1} and Claim~\ref{claim:2}, is that each binary vector is mapped to a unique element in $GF(q)$. This property also applies here because the field size is $q=2^{\ell}$. Hence, from~\eqref{e:a} we have
\begin{equation}
\label{e:ell}
Pr(F) <  \frac{t \cdot h(\delta)}{q^{c-\delta}}=\mathcal{O}\left(\frac{(k/\ell)^{\delta}}{2^{\ell(c-\delta)}} \right).
\end{equation}
From~\eqref{e:ell} we can see that the probability of decoding failure vanishes asymptotically if 
\begin{equation}
\label{e:l}
\lim_{k\rightarrow +\infty} \frac{k^{\delta}}{\ell^{\delta} 2^{\ell(c-\delta)}} =0,
\end{equation}
which holds if $\ell=\Omega(\log k)$.
\section*{Acknowledgment}
The authors would like to thank the reviewers for their valuable comments and suggestions that helped to improve the paper. The authors would also like to thank Kannan Ramchandran for valuable discussions related to an earlier draft of this work; and Ramezan Paravi Torghabeh for his help with the simulations related to the work in \cite{V15,R10}.

\bibliographystyle{IEEEtran}
\bibliography{Refs}

\begin{thebibliography}{10}
\providecommand{\url}[1]{#1}
\csname url@samestyle\endcsname
\providecommand{\newblock}{\relax}
\providecommand{\bibinfo}[2]{#2}
\providecommand{\BIBentrySTDinterwordspacing}{\spaceskip=0pt\relax}
\providecommand{\BIBentryALTinterwordstretchfactor}{4}
\providecommand{\BIBentryALTinterwordspacing}{\spaceskip=\fontdimen2\font plus
\BIBentryALTinterwordstretchfactor\fontdimen3\font minus
  \fontdimen4\font\relax}
\providecommand{\BIBforeignlanguage}[2]{{%
\expandafter\ifx\csname l@#1\endcsname\relax
\typeout{** WARNING: IEEEtran.bst: No hyphenation pattern has been}%
\typeout{** loaded for the language `#1'. Using the pattern for}%
\typeout{** the default language instead.}%
\else
\language=\csname l@#1\endcsname
\fi
#2}}
\providecommand{\BIBdecl}{\relax}
\BIBdecl

\bibitem{GC}
S.~{Kas Hanna} and S.~{El Rouayheb}, ``Guess check codes for deletions and
  synchronization,'' in \emph{2017 IEEE International Symposium on Information
  Theory (ISIT)}, June 2017, pp. 2693--2697.

\bibitem{VT65}
R.~Varshamov and G.~Tenengol’ts, ``Correction code for single asymmetric
  errors,'' \emph{Automat. Telemekh}, vol.~26, no.~2, pp. 286--290, 1965.

\bibitem{V15}
R.~Venkataramanan, V.~N. Swamy, and K.~Ramchandran, ``Low-complexity
  interactive algorithms for synchronization from deletions, insertions, and
  substitutions,'' \emph{IEEE Transactions on Information Theory}, vol.~61,
  no.~10, pp. 5670--5689, 2015.

\bibitem{R10}
R.~Venkataramanan, H.~Zhang, and K.~Ramchandran, ``Interactive low-complexity
  codes for synchronization from deletions and insertions,'' in \emph{2010 48th
  Annual Allerton Conference on Communication, Control, and Computing
  (Allerton)}, Sept 2010, pp. 1412--1419.

\bibitem{Y14}
S.~S.~T. Yazdi and L.~Dolecek, ``A deterministic polynomial-time protocol for
  synchronizing from deletions,'' \emph{IEEE Transactions on Information
  Theory}, vol.~60, no.~1, pp. 397--409, 2014.

\bibitem{S16}
F.~Sala, C.~Schoeny, N.~Bitouzé, and L.~Dolecek, ``Synchronizing files from a
  large number of insertions and deletions,'' \emph{IEEE Transactions on
  Communications}, vol.~64, no.~6, pp. 2258--2273, June 2016.

\bibitem{M11}
N.~Ma, K.~Ramchandran, and D.~Tse, ``Efficient file synchronization: A
  distributed source coding approach,'' in \emph{Information Theory Proceedings
  (ISIT), 2011 IEEE International Symposium on}.\hskip 1em plus 0.5em minus
  0.4em\relax IEEE, 2011, pp. 583--587.

\bibitem{RO}
R.~Gabrys, E.~Yaakobi, and O.~Milenkovic, ``Codes in the damerau distance for
  {DNA} storage,'' in \emph{2016 IEEE International Symposium on Information
  Theory (ISIT)}, July 2016, pp. 2644--2648.

\bibitem{B16}
J.~Brakensiek, V.~Guruswami, and S.~Zbarsky, ``Efficient low-redundancy codes
  for correcting multiple deletions,'' in \emph{Proceedings of the
  Twenty-Seventh Annual ACM-SIAM Symposium on Discrete Algorithms}.\hskip 1em
  plus 0.5em minus 0.4em\relax SIAM, 2016, pp. 1884--1892.

\bibitem{D01}
M.~C. Davey and D.~J. MacKay, ``Reliable communication over channels with
  insertions, deletions, and substitutions,'' \emph{IEEE Transactions on
  Information Theory}, vol.~47, no.~2, pp. 687--698, 2001.

\bibitem{R05}
E.~A. Ratzer, ``Marker codes for channels with insertions and deletions,'' in
  \emph{Annales des t{\'e}l{\'e}communications}, vol.~60, no. 1-2.\hskip 1em
  plus 0.5em minus 0.4em\relax Springer, 2005, pp. 29--44.

\bibitem{M06}
M.~Mitzenmacher and E.~Drinea, ``A simple lower bound for the capacity of the
  deletion channel,'' \emph{IEEE Transactions on Information Theory}, vol.~52,
  no.~10, pp. 4657--4660, Oct 2006.

\bibitem{D06}
E.~Drinea and M.~Mitzenmacher, ``On lower bounds for the capacity of deletion
  channels,'' \emph{IEEE Transactions on Information Theory}, vol.~52, no.~10,
  pp. 4648--4657, Oct 2006.

\bibitem{D07}
S.~Diggavi, M.~Mitzenmacher, and H.~D. Pfister, ``Capacity upper bounds for the
  deletion channel,'' in \emph{2007 IEEE International Symposium on Information
  Theory}, June 2007, pp. 1716--1720.

\bibitem{Ven13}
R.~Venkataramanan, S.~Tatikonda, and K.~Ramchandran, ``Achievable rates for
  channels with deletions and insertions,'' \emph{IEEE Transactions on
  Information Theory}, vol.~59, no.~11, pp. 6990--7013, Nov 2013.

\bibitem{K13}
Y.~Kanoria and A.~Montanari, ``Optimal coding for the binary deletion channel
  with small deletion probability,'' \emph{IEEE Transactions on Information
  Theory}, vol.~59, no.~10, pp. 6192--6219, Oct 2013.

\bibitem{Rah15}
M.~Rahmati and T.~M. Duman, ``Upper bounds on the capacity of deletion channels
  using channel fragmentation,'' \emph{IEEE Transactions on Information
  Theory}, vol.~61, no.~1, pp. 146--156, Jan 2015.

\bibitem{M09}
M.~Mitzenmacher, ``A survey of results for deletion channels and related
  synchronization channels,'' \emph{Probability Surveys}, vol.~6, pp. 1--33,
  2009.

\bibitem{L66}
V.~I. Levenshtein, ``Binary codes capable of correcting deletions, insertions
  and reversals,'' in \emph{Soviet physics doklady}, vol.~10, 1966, p. 707.

\bibitem{S02}
N.~J. Sloane, ``On single-deletion-correcting codes,'' \emph{Codes and Designs,
  de Gruyter, Berlin}, pp. 273--291, 2002.

\bibitem{N14}
D.~Cullina and N.~Kiyavash, ``An improvement to levenshtein's upper bound on
  the cardinality of deletion correcting codes,'' \emph{IEEE Transactions on
  Information Theory}, vol.~60, no.~7, pp. 3862--3870, July 2014.

\bibitem{H02}
A.~S. Helberg and H.~C. Ferreira, ``On multiple insertion/deletion correcting
  codes,'' \emph{IEEE Transactions on Information Theory}, vol.~48, no.~1, pp.
  305--308, 2002.

\bibitem{A12}
K.~A. Abdel-Ghaffar, F.~Paluncic, H.~C. Ferreira, and W.~A. Clarke, ``On
  helberg's generalization of the levenshtein code for multiple
  deletion/insertion error correction,'' \emph{IEEE Transactions on Information
  Theory}, vol.~58, no.~3, pp. 1804--1808, 2012.

\bibitem{S99}
L.~J. Schulman and D.~Zuckerman, ``Asymptotically good codes correcting
  insertions, deletions, and transpositions,'' \emph{IEEE transactions on
  information theory}, vol.~45, no.~7, pp. 2552--2557, 1999.

\bibitem{G14}
V.~Guruswami and C.~Wang, ``Deletion codes in the high-noise and high-rate
  regimes,'' \emph{IEEE Transactions on Information Theory}, vol.~63, no.~4,
  pp. 1961--1970, April 2017.

\bibitem{LG16}
V.~Guruswami and R.~Li, ``Efficiently decodable insertion/deletion codes for
  high-noise and high-rate regimes,'' in \emph{2016 IEEE International
  Symposium on Information Theory (ISIT)}, July 2016, pp. 620--624.

\bibitem{Ob}
------, ``Coding against deletions in oblivious and online models,''
  \emph{arXiv preprint arXiv:1612.06335}, 2016.

\bibitem{Software}
S.~{Kas Hanna} and S.~{El Rouayheb}, ``Implementation of {G}uess \& {C}heck
  ({GC}) codes,'' 2017, \url{http://eceweb1.rutgers.edu/csi/software.html}.

\end{thebibliography}

\end{document}